\documentclass[conference,a4paper]{IEEEtran}

\usepackage[utf8]{inputenc} 
\usepackage[T1]{fontenc}
\usepackage{url}              
\usepackage{cite}             

\usepackage[cmex10]{amsmath}   

\usepackage{amssymb,amsfonts,amsthm}

\usepackage{mleftright}       
\mleftright                   

\usepackage{graphicx}         
\usepackage{booktabs}        
\usepackage{bm}
\usepackage{textcomp}
\usepackage{comment}
\usepackage{subcaption}
\usepackage{xcolor}

\def\IID{\mathop{\mathrm{i.i.d}}} 

\newtheorem{theorem}{\bfseries Theorem}
\newtheorem{lemma}{\bfseries Lemma}
\newtheorem{definition}{\bfseries Definition}
\newtheorem{corollary}{\bfseries Corollary}

\newtheorem{remark}{\bfseries Remark}

\newtheorem{example}{\bfseries Example}

\newtheorem*{assumptions*}{\bfseries Assumptions}

\newcommand{\RDF}[2]{{\mathcal{R}}_{#1}(#2)}%
\newcommand{\fRDF}[3]{{\mathcal{R}}_{#1,#2}(#3)}%
\newcommand{\iRDF}[2]{{{\mathcal I}}_{#1}(#2)}%
\newcommand{\ifRDF}[3]{{{\mathcal I}}_{#1,#2}(#3)}%

\newcommand{\RDFm}[2]{{\widehat{\mathcal{R}}}_{#1}(#2)}%
\newcommand{\fRDFm}[3]{{\widehat{\mathcal{R}}}_{#1,#2}(#3)}%
\newcommand{\iRDFm}[2]{{\widehat{{\mathcal I}}}_{#1}(#2)}%
\newcommand{\ifRDFm}[3]{{\widehat{{\mathcal I}}}_{#1,#2}(#3)}%

\begin{document}

\title{Indirect Rate Distortion Functions with\\ $f$-Separable Distortion Criterion} 

\author{%
  \IEEEauthorblockN{Photios A. Stavrou^\dagger, Yanina Shkel and Marios Kountouris\dagger}
  \IEEEauthorblockA{%
    ~~~\\
    ~~~\\
    ~~~}
}

%
%
   
%
\author{%
   \IEEEauthorblockN{Photios A. Stavrou\IEEEauthorrefmark{1},
                     Yanina Shkel\IEEEauthorrefmark{2},
                    Marios Kountouris\IEEEauthorrefmark{1}}
  \IEEEauthorblockA{\IEEEauthorrefmark{1}                     Communication Systems Department, EURECOM, Sophia-Antipolis, France 
                    }
  \IEEEauthorblockA{\IEEEauthorrefmark{2}
                    School of Computer and Communication Sciences, EPFL, Lausanne, Switzerland,
                    }
\IEEEauthorblockA{\texttt{\{fotios.stavrou,marios.kountouris\}}@eurecom.fr,~\texttt{yanina.shkel}@epfl.ch}
\thanks{The work of P. Stavrou and M. Kountouris has received funding from the European Research Council (ERC) under the European Union’s Horizon 2020 Research and Innovation programme (Grant agreement No. 101003431).}
 }

\maketitle

\begin{abstract}
We consider a remote source coding problem subject to a {distortion function}. Contrary to the use of the classical separable distortion criterion, herein we consider the more general, $f$-separable distortion measure and study its implications on the characterization of the minimum achievable rates (also called $f$-separable indirect rate distortion function (iRDF)) under both excess and average distortion constraints. First, we provide a single-letter characterization of the optimal rates  subject to an excess distortion using properties of the $f$-separable distortion. Our main result is a single-letter characterization of the $f$-separable iRDF  subject to an average distortion constraint. As a consequence of the previous results, we also show a series of equalities that hold using either indirect or classical RDF under $f$-separable excess or average distortions. We corroborate our results with two application examples in which new closed-form solutions are derived, and based on these, we also recover known special cases. 
\end{abstract}


\section{Introduction}\label{sec:intro}

The mathematical analysis of the lossy source coding under a {fidelity criterion}, called rate distortion theory \cite{shannon:1959}, was {developed under the assumption} that an encoder observes an information source ${\bf x}$ with distribution $p(x)$ defined on the alphabet space ${\cal X}$, and the aim is for the decoder to reconstruct in a minimal end-to-end rate-constrained manner, its representation $\widehat{\bf x}$ defined on an alphabet $\widehat{\cal X}$ within a distortion measure $d:{\cal X}\times\widehat{\cal X}\mapsto[0,\infty)$. When the information source generates a sequence of $n$ realizations, the source sequence induces the distribution $p(x^n)$ on the Cartesian product alphabet space ${\cal X}^n$, with its reconstruction alphabet being $\widehat{\cal X}^n$. For the latter case, Shannon in \cite{shannon:1959} extended the single-letter expression of the distortion measure to the $n$-letter expression $d^n:{\cal X}^n\times\widehat{\cal X}^n\mapsto[0,\infty)$ by taking the arithmetic mean of single-letter distortions, i.e.,
\begin{align}
d^n(x^n,\widehat{x}^n)=\frac{1}{n}\sum_{i=0}^n d(x_i,\widehat{x}_i),\label{separable_dist}
\end{align}
which is often encountered as {\it separable}, {\it additive} or {\it per-letter} distortion measure.
\par A natural extension of the lossy source coding problem, called {\it indirect} or {\it remote} lossy source coding, was proposed almost fifteen year later in\cite{dobrushin:1962}. Therein the authors considered the case where the encoder observes a noisy version of the source ${\bf x}$, say ${\bf z}$, and the goal is to reconstruct $\widehat{\bf x}$ with minimal rates subject to an average distortion $d:{\cal X}\times\widehat{\cal X}\mapsto[0,\infty)$. A major result in \cite{dobrushin:1962} is that for stationary memoryless sources, the fundamental limit in the asymptotic regime corresponds to the classical lossy source coding problem with an amended average distortion constraint. Subsequently, this problem and some of its variants, e.g., non-asymptotic analysis, excess distortion measures, multi-terminal systems, were revisited by many researchers, see e.g., \cite{berger:1971,sakrison:1968,wolf-ziv:1970,witsenhausen:1980,kipnis:2015,Kostina:2016,yamamoto:1980,berger:1996,oohama:2012,eswaran:2019} and references therein.  
\par All the aforementioned efforts in \cite{berger:1971,sakrison:1968,wolf-ziv:1970,witsenhausen:1980,kipnis:2015,Kostina:2016,yamamoto:1980,berger:1996,oohama:2012,eswaran:2019}, consider separable distortion penalties. On one hand, the separability assumption is natural and quite appealing when it comes to the derivation of tractable characterizations of the fundamental trade-offs between the coding (or compressed) rate and its corresponding distortion. On the other hand, the separability assumption is very restrictive because it {only models distortion penalties that are {\it linear functions} of the single-letter distortion in the source reconstruction}. However, in real-world applications, distortion measures may be highly {\it non-linear}. To address this issue and inspired by \cite{shkel:2018}, here we consider a much broader class of distortion measures, namely, $f$-separable distortion measures.   
\par In this work, we derive the following new results: (i) a single-letter characterization of the minimal rates  subject to an excess distortion using properties of the $f$-separable distortion (see Lemma \ref{lemma:char_excess_dist}); (ii) a single-letter characterization of the $f$-separable iRDF (obtained for finite alphabets) subject to an average distortion constraint that is obtained under relatively mild regularity conditions and by making use of a strong converse theorem \cite{Kostina:2016} (see Theorem \ref{theorem:multi_char}); (iii) new series of equalities under $f$-separable excess or average distortion constraints using indirect or classical RDFs (see Corollary \ref{cor:char_excess_dist} and Theorem \ref{theorem:multi_char}); (iv) two application examples in which new analytical solutions are derived for various types of $f$-separable average distortions; we also explain how these analytical expressions recover known results as special cases (see Examples \ref{example:bms_nc}, \ref{example:bec}). 
It is worth mentioning that from (ii), we also derive the implicit solution of the optimal minimizer that achieves the characterization of the $f$-separable iRDF (see Corollary \ref{lemma:implicit_sol}). This result can be readily used to derive new Blahut-Arimoto type of algorithms \cite{arimoto:1972,blahut:1972} for a much richer class of distortion penalties.


\section{Problem Formulation}\label{sec:prob_form}

We consider a memoryless source described by the tuple $({\bf x}, {\bf z})$ with probability distribution ${p}(x,z)$ in the product alphabet space ${\cal X}\times{\cal Z}$. The remote information of the source is in ${\bf x}$ whereas ${\bf z}$ is the noisy observation at the encoder side. 
The goal is to study the remote source coding problem \cite{dobrushin:1962,wolf-ziv:1970,witsenhausen:1980} under an $f$-separable distortion measure. 

Formally, the system model (without the distortion penalties) is illustrated in Fig. \ref{fig:system_model} and can be interpreted as follows. An \emph{information source} is a sequence of $n$-length independent and identically distributed ($\IID$) RVs $({\bf x}^n,{\bf z}^n)$. The {\it encoder (E)} and the {\it decoder (D)}, are modeled by the mappings
\begin{align}
\begin{split}
&f^E: {\cal Z}^n\rightarrow{\cal W},~~g^D: {\cal W}\rightarrow\widehat{\cal X}^n
\end{split}\label{coding_functions}
\end{align}
where the index set ${\cal W}\in\{1,2,\ldots,M\}$. 
\begin{figure} [htp]
	\begin{center}
\includegraphics[width=\columnwidth]{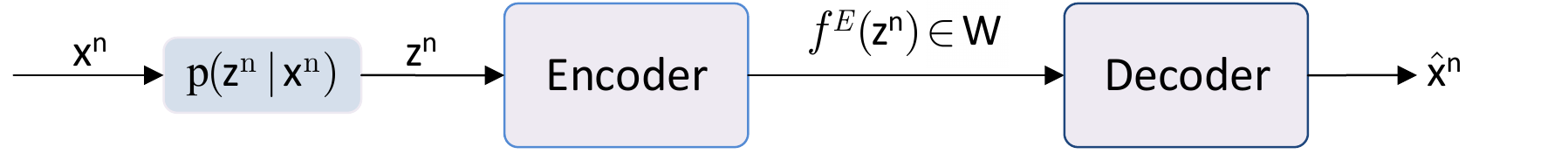}
	\end{center}
	\caption{System model.}
	\label{fig:system_model}
\end{figure}

We consider a {per-letter distortion measure} responsible to penalize the remote information source in Fig. \ref{fig:system_model} given by $d:{\cal X}\times\widehat{\cal X}\mapsto[0,\infty)$ and their corresponding $n$-letter expressions given by $d^n:{\cal X}^n\times\widehat{\cal X}^n\mapsto[0,\infty)$. {This setting has recently gained attention in the context of goal-oriented semantic communication \cite{Goldreich:ACM,stavrou:2023tcom}, where ${\bf x}$ can represent the semantic or intrinsic information of the source, which is not directly observable, whereas ${\bf z}$ is the noisy observation of the source at the encoder side. }

Next, we define the precise terminology of the noisy lossy source codes for the single-letter and the multi-letter case (without restricting to $\IID$ processes at this stage).

\begin{definition}\label{def:noisy_source_codes}(Noisy lossy source codes) Consider constants $\epsilon\in[0,1)$, $D\geq{0}$, and an integer $M$. \\ 
{\bf (1)} We say that a noisy lossy source-code $(f^E, g^D)$ is an $(M, D)$-noisy lossy source code on $({\cal X}, {\cal Z}, \widehat{\cal X}, d)$ such that ${\bf x}-{\bf z}-\widehat{\bf x}$, if ${\bf E}\left[d({\bf x},{\widehat{\bf x}})\right]\leq{D}$, where ${\widehat{\bf x}}=g^D(f^E({\bf z}))$.\\ 
{\bf (2)} We say that a noisy lossy source-code $(f^E, g^D)$ is an $(M, D, \epsilon)$-noisy lossy source code on $({\cal X}, {\cal Z}, \widehat{\cal X}, d)$ such that ${\bf x}-{\bf z}-\widehat{\bf x}$, if ${\bf P}\left[d({\bf x},{\widehat{\bf x}})>D\right]\leq{\epsilon}$ where ${\widehat{\bf x}}=g^D(f^E({\bf z}))$.\\
{\bf (3)} If $(f^E, g^D)$ is an $(M, D)$-noisy lossy source code on $({\cal X}^n, {\cal Z}^n, \widehat{\cal X}^n,  d^n)$ such that ${\bf x}^n-{\bf z}^n-\widehat{\bf x}^n$, we say that  $(f^E, g^D)$ is an $(n, M, D)$-noisy lossy source code.\\
{\bf (4)} If $(f^E, g^D)$ is an $(M, D, \epsilon)$-noisy lossy source code on $({\cal X}^n, {\cal Z}^n, \widehat{\cal X}^n, d^n)$ such that ${\bf x}^n-{\bf z}^n-\widehat{\bf x}^n$, we say that  $(f^E, g^D)$ is an $(n, M, D, \epsilon)$-noisy lossy source code.
\end{definition}
We remark the following special case of Definition \ref{def:noisy_source_codes}.
\begin{remark}(On Definition \ref{def:noisy_source_codes})
In our analysis, we will also consider as a special case the classical (noiseless) lossy source codes subject to similar single-letter and multi-letter distortion measures as in the case of noisy lossy source coding. This means that we will use special cases of Definition \ref{def:noisy_source_codes}. For example, for a noiseless lossy source code, Definition \ref{def:noisy_source_codes}, {\bf (1)}, will be modified as follows
\begin{itemize}
\item we say that a lossy source-code $(f^E, g^D)$ is an $(M, D)$-lossy source code on $({\cal X}, \widehat{\cal X}, d)$ if ${\bf E}\left[d({\bf x},{\widehat{\bf x}})\right]\leq{D}$, where ${\widehat{\bf x}}=g^D(f^E({\bf x}))$ (because ${\bf x}={\bf z}$).
\end{itemize} 
Definition \ref{def:noisy_source_codes}, {\bf (2)}-{\bf (4)}, are modified accordingly.
\end{remark}

Using \cite[Definition 1]{shkel:2018}, we consider an  \emph{f-separable distortion measure} associated with the remote information source of the setup in Fig. \ref{fig:system_model} defined as follows
\begin{align}
d_f^{n}({x}^n,\hat{x}^n)&\triangleq{f}^{-1}\left(\frac{1}{n}\sum_{i=1}^nf\left(d(x_i,\hat{x}_i)\right)\right)\label{noisy_dist_f_separable}
\end{align}
where $f(\cdot)$ is a continuous, increasing function on $[0,\infty)$.
\par In the sequel, we give the definitions of indirect and direct (or classical) RDFs under $f$-separable distortion measures. To do it, we need the following definition of achievability.
\begin{definition}\label{def:achievability} (Achievability) Suppose that a sequence of distortion measures $\{d^n:~n=1,2,\ldots\}$ on $({\cal X}^n,\widehat{\cal X}^n)$ is given, such that ${\bf x}^n-{\bf z}^n-\widehat{\bf x}^n$. Then, we define the following statements.\\
{\bf (1)} The rate distortion tuple $(R, D)$ is indirectly achievable if there exists a sequence $(n, M_n, D^n)$-noisy lossy source codes such that $\limsup_{n\rightarrow\infty}\frac{1}{n}\log{M_n} \leq{R}$, $\limsup_{n\rightarrow\infty}D^n\leq{D}$.\\
{\bf (2)} The rate distortion tuple $(R, D)$ is indirectly and excess distortion achievable if for any $\gamma>0$ there exists a sequence $(n, M_n, D+\gamma, \epsilon_n)$-noisy lossy source codes such that $\limsup_{n\rightarrow\infty}\frac{1}{n}\log{M_n} \leq{R}$, $\limsup_{n\rightarrow\infty}\epsilon_n=0$,  where $\epsilon_n$ denotes the decoding error probability, i.e., $\epsilon_n={\bf P}\left[{\bf x}^n\neq{g}^D(f^E({\bf z}^n))\right]$.
\end{definition}
If we assume sequences of noiseless lossy source codes, we say that a rate distortion tuple $(R,D)$ is directly (and excess distortion) achievable in analogous way to Definition \ref{def:achievability}, with ${\cal X}^n={\cal Z}^n$. This means that the sequence of distortion measures $\{d^n: n=1,2,\ldots\}$ can be defined either on $({\cal Z}^n,\widehat{\cal X}^n)$ or on $({\cal X}^n,\widehat{\cal X}^n)$.
\begin{definition}\label{def:iRDF}(iRDF) Given a single-letter distortion measure $d\colon {\cal X} \times \widehat{\cal X} \to [0, \infty)$ and a continuous, increasing function $f$ on $[0, \infty)$, let $\{d_f^{n}:~n=1,2,\ldots\}$ be a sequence of $f$-separable distortion measures.
Then,
\begin{align}
\ifRDF{f}{d}{D}=\inf\left\{R: (R, D)~\mbox{is indirectly achievable}\right\} \label{inf_ind_achievable_rates}
\end{align}
and
$\ifRDFm{f}{d}{D}=\inf\bigl\{R: (R, D)~\mbox{is indirectly and excess dist-}\\~
\mbox{ortion achievable}\bigr\}$.
If $f$ is the identity function, then we have a sequence of separable distortion measures; in this case we omit the subscript $f$ and write $\iRDF{d}{D}$ and $\iRDFm{d}{D}$.
\end{definition}

\begin{definition}\label{def:dRDF}(Direct RDF) Given a single-letter distortion measure $d \colon {\cal Z} \times \widehat{\cal X} \to [0, \infty)$ and a continuous, increasing function $f$ on $[0, \infty)$, let $\{d_f^{n}:~n=1,2,\ldots\}$ be a sequence of $f$-separable distortion measures.
Then,
\begin{align}
\fRDF{f}{d}{D}=\inf\left\{R: (R, D)~\mbox{is directly achievable}\right\} \label{inf_dir_achievable_rates}
\end{align}
and
$\fRDFm{f}{d}{D}=\inf\{\mbox{R: (R, D)~is directly and excess distortion}\\~\mbox{achievable}\}$.
If $f$ is the identity function, we omit the subscript $f$ and write $\RDF{d}{D}$ and $\RDFm{d}{D}$.
\end{definition}

We give the following remark for the previous two Definitions.
\begin{remark}\label{remark:2} (On Definitions \ref{def:iRDF}, \ref{def:dRDF})
In this work our goal is to characterize the $f$-separable iRDFs  $\ifRDF{d}{f}{D}$ and $\ifRDFm{d}{f}{D}$ for a given distortion measure $d(\cdot, \cdot)$ and a function $f(\cdot)$. In addition to the $f$-separable iRDFs, we consider the following three special cases: (1) separable RDF $\RDF{d}{D}$ and $\RDFm{d}{D}$,  (2) separable iRDFs $\iRDF{d}{D}$ and $\iRDFm{d}{D}$, and  (3) $f$-separable RDFs $\fRDF{d}{f}{D}$ and $\fRDFm{d}{f}{D}$.
To state our results, we compare these different classes of RDFs to each other. While the iRDFs is defined over some space $({\cal X}, {\cal Z}, \widehat{\cal X}, d)$, it is possible to generate modified direct RDFs from iRDFs in which case these are definite 
over the space $({\cal Z}, \widehat{\cal X}, {\tilde d})$, where $\tilde d \colon {\cal Z} \times \widehat{\cal X} \to [0,\infty)$ is an amended distortion measure. In general, the underlying space for the direct RDFs should be clear from context. For example, $\RDF{d}{D}$ refers to an RDF on $({\cal X}, \widehat{\cal X}, d)$, while $\RDF{\tilde d}{D}$ refers to an RDF on $({\cal Z}, \widehat{\cal X}, {\tilde d})$.
\end{remark}

\section{Prior Work} \label{sec:prior_work}

\par Next, we discuss more extensively some prior results that will be used in our main results.

\subsection{RDF under Average and Excess Constraints}
For $\IID$ sources with finite alphabets $({\cal X}, \widehat{\cal X})$ and bounded distortion measure $d$, the RDF is given by
\begin{align}
\RDF{d}{D} & =\inf_{\substack{q(\widehat{x}|x) \colon
{\bf E}\left[d({\bf x}, \widehat{\bf x})\right]\leq D
}} I({\bf x};\widehat{\bf x}).\nonumber
\end{align}
See~e.g., \cite[Theorem 10.2.1]{cover-thomas:2006} and ~\cite[Theorem 5.2.1]{han:2003}. Moreover, we know that for stationary ergodic sources with a bounded distortion measure,
\begin{align}
\RDF{d}{D} = \RDFm{d}{D}. \label{eq:rdf_eqv}
\end{align}
That is, the RDF is the same under average and excess distortion constraints~\cite[Theorem 5.9.1]{han:2003}. We also know that for stationary ergodic sources $\RDFm{d}{D}$ satisfies the so-called {\em strong converse} ~\cite{Kostina:2016,kieffer:1991}. Finally, the second order asymptotic expansion of $\RDFm{d}{D}$ is given as well, see e.g.,~\cite{Wang:2011,Kostina:2012}, but this type of analysis is beyond the scope of the present paper.

\subsection{iRDF}

For $\IID$ sources with finite alphabets $({\cal X}, {\cal Z}, \widehat{\cal X})$ and bounded distortion measure $d$, the iRDF is given by
\begin{align}
\iRDF{d}{D} &= \inf_{\substack{q(\widehat{x}|z) \colon\\
{\bf E}\left[{d}({\bf x}, \widehat{\bf x})\right]\leq D
}} I({\bf z};\widehat{\bf x})\nonumber\\
&\stackrel{(a)}=\inf_{\substack{q(\widehat{x}|z) \colon\\
{\bf E}\left[\tilde{d}({\bf z}, \widehat{\bf x})\right]\leq D
}} I({\bf z};\widehat{\bf x})\equiv\RDF{\tilde d}{D} \label{eq:irdf_main}
\end{align}
where $(a)$ follows from \cite{dobrushin:1962} (see also Remark \ref{remark:2}) and $\RDF{\tilde d}{D}$ is the direct RDF for $({\cal X}, {\cal Z})$ with the amended distortion given by $\tilde{d}({z}, \widehat{x})=\sum_{{\cal X}}p(x|z)d(x,\widehat{x})$. In other words, the indirect rate distortion problem reduces to a  direct rate distortion problem with a modified per-letter distortion measure~\cite{dobrushin:1962,wolf-ziv:1970,witsenhausen:1980,Kostina:2016}. Moreover, for $\IID$ sources, the iRDF is the same under average and excess distortion constraints
\begin{align}
\iRDF{d}{D} &= \iRDFm{d}{D} \label{eq:irdf_eqv}
\end{align}
and the strong converse also holds~\cite{Kostina:2016}. Finally, for this problem, the second-order asymptotic analysis has been addressed in~\cite{Kostina:2016} where it was shown that the equivalence between direct and indirect problems no longer holds in the second-order (dispersion) sense.

\subsection{f-Separable RDF}

Similar equivalence results hold for $f$-separable RDFs. Specifically, for $\IID$ sources
\begin{align}
\fRDF{f}{d}{D} &= \RDF{\bar d}{f(D)} =\inf_{\substack{q(\widehat{x}|x)\\
{\bf E}\left[{\bar d}({\bf x}, \widehat{\bf x})\right]\leq f(D)
}} I({\bf x};\widehat{\bf x})
\end{align}
where $\RDF{\bar d}{\cdot}$ is the separable RDF for $({\cal X}, {\cal \widehat{X}})$ with the amended distortion given by ${\bar d}(x, \widehat{x})=f(d(x, \widehat{x}))$, see~\cite{shkel:2018}. More generally, it is shown in~\cite{shkel:2018} that for the $f$-separable rate distortion problem 
\begin{align}
\fRDFm{f}{d}{D}  = \RDFm{\bar d}{f(D)}. \label{eq:fsep_avg_max}
\end{align}
That is, under excess distortion criterion, the $f$-separable RDF reduces to the classical separable case without any assumption on the underlying source. In fact for stationary ergodic sources, this result extends to both average and excess distortion criteria under some regularity assumptions (see~\cite[Theorem 1]{shkel:2018}), namely, 
\begin{align}
\fRDF{f}{d}{D}  = \fRDFm{f}{d}{D}.  \label{eq:frdf_eqv}
\end{align}
{We remark that the generalizations of the classical rate-distortion problem to indirect and $f$-separable rate distortion problems have intriguing parallels. Both generalizations could be expressed in terms of a classical amended rate distortion problem. The same insight holds when we apply both generalizations simultaneously. As we will see next, the resulting rate-distortion function could be expressed in terms of the classical amended rate distortion problem.}

\section{Single-letter characterization of the operational rates for $\IID$ sources}\label{sec:characterization}

In this section, we characterize the $f$-separable iRDFs for the setup in Fig. \ref{fig:system_model} for $\IID$ sources. Specifically, our main result states that for $\IID$ sources over finite alphabets (under mild regularity assumptions) we have that
\begin{align}
\ifRDF{f}{d}{D} &= \RDF{\tilde d}{f(D)} 
\end{align}
where $\RDF{\tilde d}{D}$ is the RDF for $({\cal X}, {\cal Z}, {\tilde d})$ with the amended distortion given by ${\tilde d}({z}, \widehat{x})=\sum_{{\cal X}}p(x|z)f(d(x,\widehat{x}))$.

\par First, we give a lemma in which we characterize the $f$-separable iRDF under the excess distortion criterion.
\begin{lemma}\label{lemma:char_excess_dist}($f$-separable iRDF under excess distortion) Given a single-letter distortion measure $d\colon {\cal X} \times \widehat{\cal X} \to [0, \infty)$ and a continuous, increasing function $f$ on $[0, \infty)$,
\begin{align}
\ifRDFm{f}{d}{D} = \iRDFm{\bar d}{f(D)} \label{char_excesss_dist}
\end{align}
where $\iRDFm{\bar d}{f(D)}$ is computed subject to the single-letter separable distortion measure $\bar{d}(x,\widehat{x})=f(d(x,\widehat{x}))$. 
\end{lemma}
Next, we make assumptions that will be used to derive the single-letter information theoretic characterization to our problem. These assumptions are a counterpart of the assumptions utilized in \cite[Theorem 1]{shkel:2018}; however, due to the difficulty of the indirect rate distortion problem, these assumptions are more restrictive, e.g., we only consider finite alphabets.
\begin{assumptions*}\label{assumptions:assumptions} Suppose that the following statements are true.
\begin{enumerate}
\item[{\bf (A1)}] The joint process $\{({\bf x}^n,{\bf z}^n):~n=1,2,\ldots\}$ is $\IID$ sequence of random variables, namely, ${p}(x^n,z^n)={p}(x){p}(z|x)\times\ldots\times{p}(x){p}(z|x)=p(x|z)p(z)\times\ldots\times{p(x|z)p(z)},$ for any $n$;
\item[{\bf (A2)}] 
The single-letter distortion ${d}(\cdot,\cdot)$ and is such that 
\begin{align}
\max_{(x,\widehat{x})\in{\cal X} \times \widehat{\cal X}} d(x,\widehat{x})<\infty; \label{assumption:eq.1}
\end{align}
\item[{\bf (A3)}] The alphabets $({\cal X}, {\cal Z}, \widehat{\cal X})$ are finite.
\end{enumerate}
\end{assumptions*}
{In particular, assumption (A2) rules out pathological rate-distortion function for which finite distortion is only possible at full rate.}
\begin{corollary}(Consequence of Lemma~\ref{lemma:char_excess_dist}) \label{cor:char_excess_dist}
Under Assumptions {\bf (A1)}-{\bf (A3)}, a consequence of Lemma \ref{lemma:char_excess_dist} is the following series of equalities 
\begin{align}
\ifRDFm{f}{d}{D} = \iRDFm{\bar d}{f(D)} = \fRDFm{f}{\hat d}{D} = \RDFm{\tilde d}{f(D)}
\end{align}
where 
\begin{align}
\bar{d}({x}, \widehat{x})&=f(d(x,\widehat{x}))\label{f_separable_dist}\\
\hat{d}(z,\widehat{x})&=f^{-1} \left(\sum_{{x}}p(x|z)f(d(x,\widehat{x}))\right)\label{inv_amended_f_separable_dist}\\
{\tilde d} (z, \widehat{x}) &= \sum_{{x}}p(x|z)f(d(x,\widehat{x})).\label{amended_f_separable}
\end{align}
\end{corollary}
\begin{IEEEproof}
The first equality, $\ifRDFm{f}{d}{D} = \iRDFm{\bar d}{f(D)}$, is shown in Lemma~\ref{lemma:char_excess_dist}. We have that $\iRDFm{\bar d}{f(D)} = \RDFm{\tilde d}{f(D)}$ from~(\ref{eq:rdf_eqv}),~(\ref{eq:irdf_main}) and~(\ref{eq:irdf_eqv}). Finally, $\fRDFm{f}{\hat d}{D} = \RDFm{\tilde d}{f(D)}$ follows from~(\ref{eq:fsep_avg_max}). This completes the proof.
\end{IEEEproof}

Next, we show the same result for the average rate-distortion functions.
\begin{theorem}\label{theorem:multi_char}($f$-separable iRDF under average distortion) Under Assumptions {\bf (A1)}-{\bf (A3)}, the $f$-separable iRDF under an average distortion constraint satisfies the following equality
\begin{align}
\ifRDF{f}{d}{D} &= \iRDF{\bar d}{f(D)} \label{eq:thm_eq1}
\end{align}
where $\bar{d}({x}, \widehat{x})$ is given in \eqref{f_separable_dist}. In particular, this implies that under Assumptions {\bf (A1)}-{\bf (A3)},  
\begin{align}
\ifRDF{f}{d}{D} &= \ifRDFm{f}{d}{D} = \fRDF{f}{\hat d}{D} = \RDF{\tilde d}{f(D)}  \label{eq:thm_eq2a}
\end{align}
and
\begin{align}
\ifRDF{f}{d}{D} & = \inf_{\substack{q(\widehat{x}|z)\\
{\bf E}\left[{\tilde d}({\bf z}, \widehat{\bf x})\right]\leq{f}({D})}}I({\bf z};\widehat{\bf x})\label{eq:thm_eq2b}
\end{align}
where $\hat{d}(z,\widehat{x})$ and ${\bar d}({z}, \widehat{x})$ are given by  \eqref{inv_amended_f_separable_dist} and \eqref{amended_f_separable}, respectively.
\end{theorem}
\begin{IEEEproof}
Equations~(\ref{eq:thm_eq2a}) and~(\ref{eq:thm_eq2b}) follow from~(\ref{eq:thm_eq1}) and the results in Section~\ref{sec:prior_work}. Namely, we have that $\iRDF{\bar d}{f(D)} = \RDF{\tilde d}{f(D)}$ from~(\ref{eq:irdf_main}); $ \fRDF{f}{\hat d}{D} = \RDF{\tilde d}{f(D)}$ from~(\ref{eq:fsep_avg_max}) and~(\ref{eq:frdf_eqv}), and $\iRDF{\bar d}{f(D)} = \iRDFm{\bar d}{f(D)} = \ifRDFm{f}{d}{D}$ from~(\ref{eq:irdf_eqv}) and Lemma~\ref{lemma:char_excess_dist}. Likewise,~(\ref{eq:thm_eq2b}) is a consequence of~(\ref{eq:thm_eq1}) and~(\ref{eq:irdf_main}).\\
It remains to show~(\ref{eq:thm_eq1}). To do it, we need the following useful lemma.
\begin{lemma}\label{lemma:A1} Suppose that the remote source $({\bf x}^n, {\bf z}^n)$ and the sequence of distortion measures $\{d^n\}_{n=1}^\infty$ are such that
\begin{align}
\limsup_{n \to \infty} \sup_{(x^n, \widehat{x}^n)}d^n( x^n, \widehat {x}^n) \leq \Delta < \infty.\label{cond_lemma_a1}
\end{align}
Then, if the rate-distortion pair $(R,D)$ is excess distortion achievable, it is achievable under the average distortion. 
\end{lemma}
First note that $f$-separable iRDF can be upper bounded as follows:
\begin{align}
\ifRDF{f}{d}{D} \stackrel{(a)} \leq \ifRDFm{f}{d}{D} \stackrel{(b)}= \iRDFm{\bar d}{f(D)} \stackrel{(c)}= \iRDF{\bar d}{f(D)}\label{char_achievable_rates}    
\end{align}
where $(a)$ is a consequence of Assumption {\bf (A2)} and Lemma \ref{lemma:A1}; $(b)$ follows from Lemma \ref{lemma:char_excess_dist}; $(c)$ follows from the equivalence between excess and average iRDF, see~(\ref{eq:irdf_eqv}).\\
The other direction,
\begin{align}
\ifRDF{f}{d}{D} \geq \iRDF{\bar d}{f(D)}
\end{align}
is a consequence of the strong converse by~\cite{Kostina:2016}. This completes the proof.
\end{IEEEproof}
One pleasing consequence of Theorem \ref{theorem:multi_char} is the following corollary.
\begin{corollary}\label{lemma:implicit_sol}(Implicit solution of $\iRDF{\bar d}{f(D)}$) The characterization in \eqref{eq:thm_eq2b} via \eqref{eq:thm_eq1} admits the following implicit solution to its minimizer 
\begin{align}
p^*(\widehat{x}|z)=\frac{e^{s\tilde{d}(z,\widehat{x})}p^*(\widehat{x})}{\sum_{\widehat{x}}e^{s\tilde{d}(z,\widehat{x})}p^*(\widehat{x})}\label{optimal_minimizer},
\end{align}
where $s<0$ is the Lagrange multiplier associated with the amended distortion penalty ${\bf E}[{\tilde d}({z}, \widehat{x})]\leq{f}(D)$ and $p^*(\widehat{x})=\sum_{z}q^*(\widehat{x}|z)p(z)$ is the $\widehat{\cal X}$-marginal of the output $\IID$ process $\widehat{\bf x}^n$. Moreover, the optimal parametric solution of  \eqref{eq:thm_eq2b} via \eqref{eq:thm_eq2a} when $\ifRDF{f}{d}{D}>0$ is given by
\begin{align}
&\ifRDF{f}{d}{D^*}=sf(D^*)-\sum_{z}p(z)\log\left(\sum_{\widehat{x}}e^{s\tilde{d}({z},\widehat{x})}p^*(\widehat{x})\right).\label{optimal_parametric_sol}
\end{align}
\end{corollary}
{By taking $p(z|x)$ to be a noiseless channel, Corollary~\ref{lemma:implicit_sol} gives us an implicit solution for $\fRDF{f}{d}{D}$ which was suggested in~\cite{shkel:2018}.}

\section{Examples}\label{sec:simulations}

In what follows, we give two examples to demonstrate the impact of $f$-separable distortion measures to a popular class of finite alphabet sources. 

\begin{example}(Binary memoryless sources)\label{example:bms_nc} Let the joint process $({\bf x}^n,{\bf z}^n)$ form an $\IID$ sequence of RVs such that ${\cal X}={\cal Z}=\widehat{\cal X}=\{0,1\}$ furnished with the classical single-letter Hamming distortion, i.e.,
\begin{align}
d(x,\widehat{x})=\begin{cases}
0,~&\mbox{if $x=\widehat{x}$}\\
1~&\mbox{if $x\neq\widehat{x}$}.
\end{cases}\label{hamming_dist}    
\end{align}
Moreover, let ${\bf x}_i\sim Bernoulli(\frac{1}{2})$  and a binary memoryless channel that induces a transition probability of the form
\begin{align}
p(z|x)=\begin{bmatrix}
1-\beta & \beta\\
\beta & 1-\beta
\end{bmatrix},~~\beta\in\left[0,\frac{1}{2}\right).\label{trans_prob_matrix}
\end{align}
\end{example}
Using the above input data, we obtain the following theorem.
\begin{theorem}(Closed-form solution)\label{thm:closed_form_binary_hamming}
For the previous inputs and for any continuous, increasing function $f(\cdot)$, we obtain
\begin{align}
&\ifRDF{f}{d}{D}=\iRDF{\bar d}{f(D)}=\nonumber\\
&\left[1-h_b\left(\frac{{f}(D)-(1-\beta)f(0)-\beta f(1)}{(1-\beta)f(1)+\beta f(0)-(1-\beta)f(0)-\beta f(1)}\right)\right]^{+}
\label{general_closed_form_bes_noisy_channel}    
\end{align}
where $[\cdot]^{+}=\max\{0,\cdot\}$, $f(D)\in\left[(1-\beta)f(0)+\beta f(1),\frac{f(0)+f(1)}{2}\right]$ and $h_b(\cdot)$ denotes the binary entropy function.
\end{theorem}

In Fig.~\ref{fig:f_separable_noisy_rdf} we illustrate some plots of  \eqref{general_closed_form_bes_noisy_channel} for various functions $f(\cdot)$ and different distortion levels $D$. It should be noted that due to the nature of the indirect rate distortion problem compared to the classical rate distortion problem, there are different minimum distortion thresholds for which the curves are well-defined. In particular, when the function $f$ is exponential, with $\beta=0.01$ and $\rho=9.2$, Fig.~\ref{fig:f_separable_noisy_rdf} demonstrates that the $f$-separable iRDF curve is non-convex, monotonic and well-defined for $D\in\left(D^{\exp}_{\min}, D_{\max}^{\exp}\right]=\left(\frac{1}{\rho}\log(1-\beta+\beta\exp(\rho)), \frac{1}{\rho}\log\left(\frac{1+\exp(\rho)}{2}\right)\right]$. Similarly, if the function $f$ is third order polynomial with $\beta=0.15$ and $\alpha=0.4$ or quadratic with $\beta=0.001$, then, from Fig. \ref{fig:f_separable_noisy_rdf} we observe that  $\iRDF{f, d}{D}$ is again non-convex, monotonic and well-defined for $D\in\left(D^{pol}_{\min}, D^{pol}_{\max}\right]=\left(\sqrt[3]{(1-a)^3\beta-a^3(1-b)}+a,\sqrt[3]{\frac{(1-a)^3-a^3}{2}}+a\right]$ and for $D\in\left(D^{qua}_{\min}, D^{qua}_{\max}\right]=\left(\sqrt{\beta},\sqrt{\frac{1}{2}}\right]$, respectively. 
Clearly, if in Fig. \ref{fig:f_separable_noisy_rdf} we consider the function $f$ to be the identity map, then, as Fig. \ref{fig:f_separable_noisy_rdf} demonstrates, we obtain $\iRDF{f, d}{D}=\iRDF{\bar d}{f(D)}={\cal R}_{\tilde d}(D)$ and the closed-form solution of \eqref{general_closed_form_bes_noisy_channel} recovers the solution of  \cite[Exercise 3.8]{berger:1971} i.e.,
\begin{align}
\iRDF{f, d}{D}=\left[1-h_b\left(\frac{D-\beta}{1-2\beta}\right)\right]^{+} \mbox{if $D\in[\beta,\frac{1}{2}]$}.\label{closed_form_ind_rdf}
\end{align}
This example aims at further emphasizing on the impact of the $f$-separable (non-linear) distortion constraint on the indirect rate distortion curve as opposed to the classical  separable (linear) distortions for which the indirect rate-distortion curve is always convex.
\begin{figure}[htp]
\centering
\includegraphics[width=\linewidth]{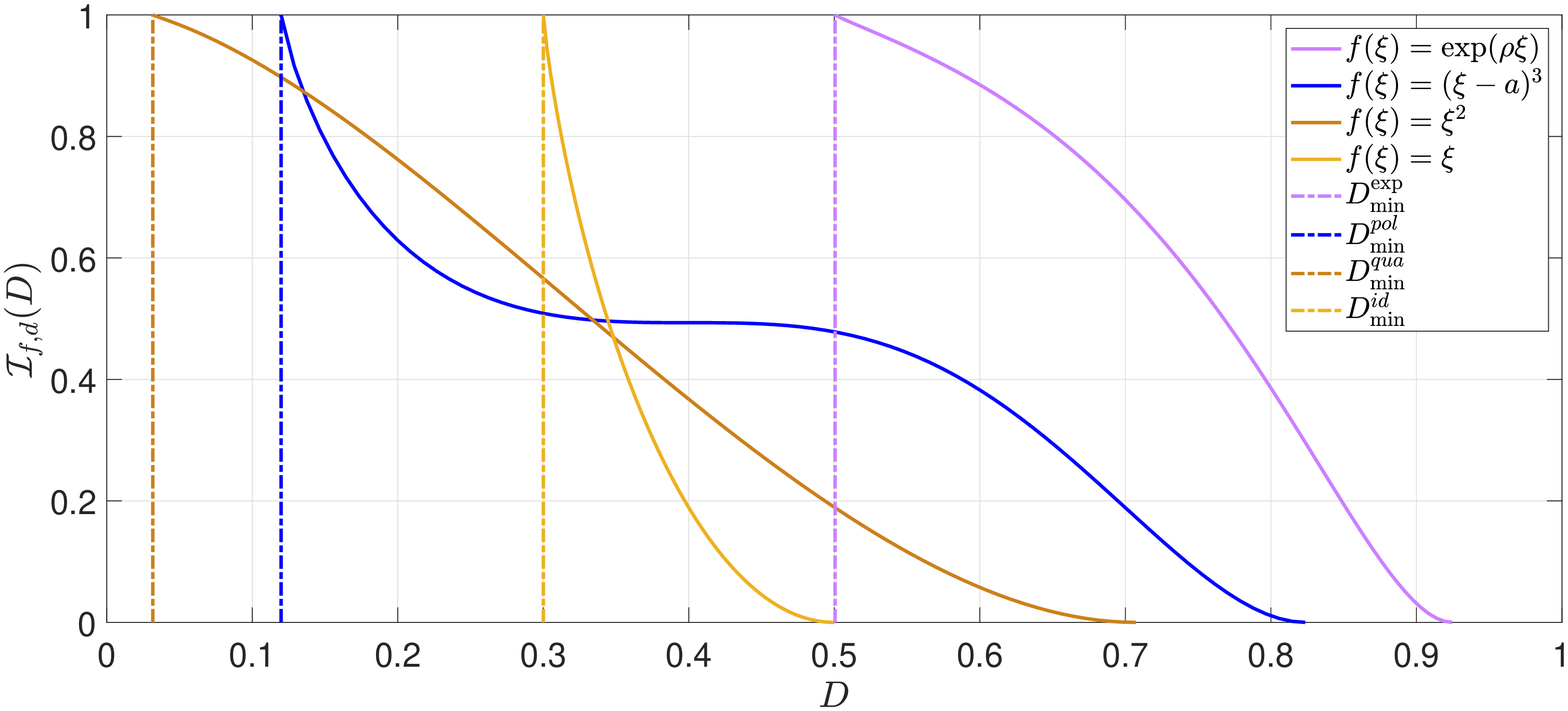}
\caption{Computation of $\ifRDF{f}{d}{D}$ for various functions $f(\cdot)$ and single-letter Hamming distance.}
\label{fig:f_separable_noisy_rdf}
\end{figure}
\paragraph*{Special case} If in Example \ref{example:bms_nc} we assume that in \eqref{trans_prob_matrix} we have $\beta=0$, then our problem recovers the solution of \cite[eq. (44)]{shkel:2018} for ${\bf x}\sim Bernoulli(\frac{1}{2})$.

\begin{example}\label{example:bec}
Let the joint process $({\bf x}^n,{\bf z}^n)$ form an $\IID$ sequence of RVs such that ${\cal X}=\widehat{\cal X}=\{0,1\}$, ${\cal Z}=\{0,e,1\}$ furnished with the Hamming distortion in \eqref{hamming_dist}. Moreover, let ${\bf x}_i\sim Bernoulli(\frac{1}{2})$  and a binary memoryless erasure channel that induces a transition probability of the form
\begin{align}
p(z|x)=\begin{bmatrix}
1-\delta & 0\\
\delta & \delta\\
0 & 1-\delta
\end{bmatrix},~~\delta\in\left[0,1\right].\label{trans_prob_matrix_bec}
\end{align}
\end{example}
Using the above input data, we obtain the following theorem.
\begin{theorem}(Closed-form solution)
For the previous input data, and for any continuous, increasing function $f(\cdot)$ we obtain
\begin{align}
&\ifRDF{f}{d}{D}=\iRDF{\bar d}{f(D)}=\nonumber\\
&\left[
(1-\delta)\left(\log(2)-h_b\left(\frac{{f}(D)-\frac{\delta}{2}f(1)-f(0)(1-\frac{\delta}{2})}{(1-\delta)(f(1)-f(0))}\right)\right)\right]^{+}\label{general_closed_form_bec}
\end{align}
where ${f}(D)\in\left[(1-\frac{\delta}{2})f(0)+\frac{\delta}{2} f(1),\frac{f(1)+f(0)}{2}\right]$.
\end{theorem}
\paragraph*{Special case} If the chosen $f$-separable distortion measure is {additive} 
(function $f$ corresponds to the identity map), then the closed-form solution of \eqref{general_closed_form_bec} recovers the solution of \cite[Eq. (76)]{Kostina:2016}, which in turn admits the closed-form solution
\begin{align}
\iRDF{f, d}{D}=\left[(1-\delta)\left(\log(2)-h_b\left(\frac{D-\frac{\delta}{2}}{1-\delta}\right)\right)\right]^{+}\label{closed_form_bec_classical}
\end{align}
where $D\in\left[\frac{\delta}{2},\frac{1}{2}\right]$.

\section*{Acknowledgement}
The work of P. A. Stavrou and M. Kountouris has received funding from the European Research Council (ERC) under the European Union’s Horizon 2020 Research and Innovation programme (Grant agreement No. 101003431).
\newpage

\IEEEtriggeratref{11}

\bibliographystyle{IEEEtran}

\bibliography{string,references}

\end{document}